\documentclass[runningheads,a4paper]{llncs}

\usepackage{amsmath, amssymb} 
\usepackage{verbatim} 
\usepackage[pdftex]{graphicx}
\usepackage{setspace}
\usepackage{algorithmic}
\usepackage{algorithm} 
\usepackage{subfig}
\usepackage{times}
\usepackage{paralist}
\usepackage{cases}
\usepackage{cancel}
\usepackage{cite}

\newcommand{\ditto}{ $-$ }

\usepackage{paralist}
\usepackage{cases}

\usepackage{tabularx}

\usepackage{tikz}

\usepackage{color}

\newcommand{\ceiling}[1]{\left\lceil{#1}\right\rceil}
\newcommand{\floor}[1]{\left\lfloor{#1}\right\rfloor}
\newcommand{\setof}[1]{\left\{{#1}\right\}}

\newcommand{\kuframework}[1]{$\mathbf{k^2U}$}

 \def\myendproof{{\ \vbox{\hrule\hbox{%
   \vrule height1.3ex\hskip0.8ex\vrule}\hrule }}\par}
\renewenvironment{proof}{\noindent{\bf Proof. }}{\hfill\myendproof}
\title{A Note on the Exact Schedulability Analysis for Segmented
  Self-Suspending Systems\thanks{This report has been supported by DFG, as
    part of the Collaborative Research Center SFB876
    (http://sfb876.tu-dortmund.de/).}}

\author{
    Jian-Jia Chen}
\institute{
    Department of Informatics,
    TU Dortmund University, Germany\\
}
\begin{document}

\maketitle

\vspace{-0.2in}
\begin{abstract}
  This report considers a sporadic real-time task system with $n$
  sporadic tasks on a uniprocessor platform, in which the
  lowest-priority task is a segmented self-suspension task and the
  other higher-priority tasks are ordinary sporadic real-time
  tasks. This report proves that the schedulability analysis for
  fixed-priority preemptive scheduling even with only one segmented
  self-suspending task as the lowest-priority task is $co{\cal
    NP}$-hard in the strong sense. Under fixed-priority preemptive
  scheduling, Nelissen et al. in ECRTS 2015 provided a mixed-integer
  linear programming (MILP) formulation to calculate an upper bound on
  the worst-case response time of the lowest-priority self-suspending
  task.  This report provides a comprehensive support to explain
  several hidden properties that were not provided in their
  paper. We also provide an input task set to explain why the
  resulting solution of their MILP formulation can be quite far from
  the exact worst-case response time.
% The formal proof is also used
% by the author to argue that the problem is in fact co${\cal
%   NP}$-hard in the strong sense.
\end{abstract}

\section{Introduction and Models}

We consider a system ${\bf T}$ of $n$ sporadic real-time tasks. A
sporadic task $\tau_i$ in ${\bf T}$ releases an infinite number of
jobs that arrive with the minimum inter-arrival time constraint. A
sporadic real-time task $\tau_i$ is characterized by its
\emph{worst-case execution time} $C_i$, its \emph{minimum
  inter-arrival time} (also called period) $T_i$ and its \emph{relative deadline} $D_i$.
In addition, each job of task $\tau_i$ has also a specified worst-case
self-suspension time $S_i$.  When a job of task $\tau_i$ arrives at
time $t$, the job should finish no later than its \emph{absolute
  deadline} $t+D_i$, and the next job of task $\tau_i$ can only be
released no earlier than $t+T_i$.  If the relative deadline $D_i$ of
task $\tau_i$ in the task set is always equal to (no more than,
respectively) the period $T_i$, such a task set is called a
\emph{implicit-deadline} (\emph{constrained-deadline}, respectively)
task set (system). If $D_i > T_i$ for a certain task $\tau_i$ in ${\bf
  T}$, then the task system is an \emph{arbitrary-deadline} task
system.  The response time of a
job is defined as its finishing time minus its release (arrival)
time. The worst-case response time $WCRT_i$ is 
the upper bound on the response times of all the jobs of task $\tau_i$.
The \emph{response time analysis} of a task  $\tau_i$ under a scheduling
algorithm is to provide a safe upper bound on $WCRT_i$.

There are two typical models for self-suspending sporadic task
systems: 1) the dynamic self-suspension task model, and 2) the
segmented self-suspension task model. In the \emph{dynamic}
self-suspension task model, e.g.,~\cite{ECRTS-AudsleyB04,RTAS-AudsleyB04,RTCSA-KimCPKH95,MingLiRTCSA1994,LiuChen:rtss2014,Huang_2015,ChenNelissenHuang-ecrts16}, in addition to the worst-case execution time
$C_i$ of sporadic task $\tau_i$, we have also the worst-case
self-suspension time $S_i$ of task $\tau_i$. In the \emph{segmented} self-suspension
task model, e.g., \cite{RTCSA-BletsasA05,Kim2016,WC16-suspend-DATE,RTSS-ChenL14,Huang:multiseg,PH:rtss98}, the execution behaviour of a job of task $\tau_i$ is
specified by interleaved computation segments and self-suspension
intervals.  The dynamic
self-suspension model provides a simple specification by ignoring the
juncture of I/O access, computation offloading, or
synchronization. However, if the suspending behaviour can be
characterized by using a segmented pattern, the segmented
self-suspension task model can be more appropriate.

This report considers a \emph{segmented} self-suspension task
model. If a task $\tau_i$ can suspend itself, a job of task $\tau_i$
is further characterized by the computation segments and suspension
intervals as an array
$(C_{i}^1,S_{i}^1,C_{i}^2,S_{i}^2,...,S_{i}^{m_i-1},C_{i}^{m_i})$,
composed of $m_i$ computation segments separated by $m_i-1$ suspension
intervals. 

In this report, we will \emph{only} consider the following special
case in fixed-priority (FP) preemptive scheduling:
\begin{itemize}
\item Task $\tau_n$ is the lowest-priority task and is a segmented
  self-suspension task. We will further assume that $D_n \leq T_n$.
\item There are $n-1$ higher-priority tasks, $\tau_1,
  \tau_2, \ldots, \tau_{n-1}$. These $n-1$ tasks are indexed from the
  highest priority $\tau_1$ to the lowest priority $\tau_{n-1}$.  The
  task set $\setof{\tau_1, \tau_2, \ldots, \tau_{n-1}}$ can be an
  arbitrary-, constrained-, or implicit-deadline task
  set.
\end{itemize}

Since we only have one self-suspending task in this report, we use $m$
to denote $m_n$ for notational simplicity, where $m \geq 2$. Moreover, the arrival time
of a computation segment is defined as the moment when the computation
segment is ready to be executed, after all its previous computation
segments and self-suspension intervals are done. In this report,
the response time $R_j$ of a computation segment $C_n^j$ is defined as
the finishing time of the computation segment minus the arrival time
of the computation segment.

We consider uniprocessor fixed-priority preemptive scheduling.  We say
that a release pattern of the tasks in ${\bf T}$ is \emph{valid} if
the jobs of the tasks in ${\bf T}$ do not violate any of the temporal
characteristics regarding to the minimum inter-arrival time, worst-case execution
time, and worst-case self-suspension time.  We say that a schedule is
\emph{feasible} if all the deadlines are met for a valid release
pattern of the tasks in ${\bf T}$. Moreover, a task system (set) is
\emph{schedulable} by a scheduling algorithm if the resulting schedule for any valid release pattern of ${\bf T}$
is always feasible. A \emph{schedulability test} of a scheduling
algorithm for a given task system is to validate whether the task
system is schedulable by the scheduling algorithm. A \emph{sufficient}
schedulability test provides only sufficient conditions for validating
the schedulability of a task system.  A \emph{necessary}
schedulability test provides only necessary conditions to allow the
schedulability of a task system.  An \emph{exact} schedulability test
provides necessary and sufficient conditions for validating the
schedulability.

By the assumption of fixed-priority preemptive scheduling, a
schedulability test of task $\tau_i$ can be done by removing all the
lower-priority tasks, $\tau_{i+1}, \ldots, \tau_n$. Since task
$\tau_n$ is the lowest-priority task, the schedulability test of the
$n-1$ higher-priority tasks under the given FP preemptive scheduling
can be done by using the well-known response time analysis,
\cite{DBLP:conf/rtss/Lehoczky90,DBLP:conf/rtss/LehoczkySD89}.

Therefore, the remaining problem is to validate whether the segmented
self-suspension task $\tau_n$ is schedulable by FP preemptive
scheduling (as the lowest-priority task). For this problem, the only
results in the literature were provided by Lakshmanan and Rajkumar
\cite{LR:rtas10} and Nelissen et al. \cite{ecrts15nelissen}.
Lakshmanan and Rajkumar proposed a pseudo-polynomial-time worst-case
response time analysis, by revising the well-known critical instant theorem originally defined in \cite{Liu_1973}. This has been recently disproved by Nelissen et
al. \cite{ecrts15nelissen}. The schedulability test by Nelissen et
al. \cite{ecrts15nelissen} requires exponential-time complexity even
for such a case when the task system has \emph{only one
  self-suspending task}.  Furthermore, Nelissen et
al. \cite{ecrts15nelissen} also assumed that all the tasks are with
constrained deadlines, i.e., $D_i \leq T_i$ for every task $\tau_i \in
{\bf T}$.  The other solutions \cite{Huang:multiseg,PH:rtss98}
require pseudo-polynomial time complexity but are only sufficient
schedulability tests.

Regarding to computational complexity, it was shown by Ridouard et
al. \cite{DBLP:conf/rtss/RidouardRC04} that the scheduler design
problem for the segmented self-suspension task model is ${\cal
  NP}$-hard in the strong sense.\footnote{Ridouard et
  al. \cite{DBLP:conf/rtss/RidouardRC04} termed this problem as the
  feasibility problem for the decision version to verify the existence
  of a feasible schedule.} The proof in
\cite{DBLP:conf/rtss/RidouardRC04} only needs each segmented
self-suspending task to have one self-suspension interval with two
computation segments.  It was reported by Chen et
al. \cite{suspension-review-jj} that the schedulability test problem
in several cases is also strongly $co{\cal NP}$-hard under
dynamic-priority scheduling, in which the priority of a job may change
over time. Such observations made in \cite{suspension-review-jj} are
based on the special cases to reduce from the schedulability test
problem of the ordinary constrained-deadline sporadic task systems
(without self-suspension), which has been recently proved to be
$co{\cal NP}$-hard in the strong sense by Ekberg and Wang
\cite{DBLP:conf/ecrts/Ekberg015} under earliest-deadline-first (EDF)
scheduling.  For fixed-priority (FP) preemptive scheduling, in which a
task is assigned a fixed priority level, the computational complexity
of the schedulability test problem for the segmented self-suspension
task model is open.

{\bf Contribution:} This report provides the following results of the
schedulability test problem and the worst-case response time analysis
for self-suspending sporadic task systems:
\begin{itemize}
\item In Section~\ref{sec:computational-complexity-fp}, we prove that
  the schedulability analysis for fixed-priority (FP)
  preemptive scheduling even with only one segmented self-suspending
  task as the lowest-priority task is $co{\cal NP}$-hard in the strong
  sense when there are more than one self-suspension interval (or
  equivalently more than two computation segments). The computational
  complexity analysis is valid for both implicit-deadline and
  constrained-deadline cases, when the priority assignment is given.
  Our proof also shows that validating whether there exists a feasible
  priority assignment is also $co{\cal NP}$-hard in the strong sense
  for constrained-deadline segmented self-suspending task systems.
\item This report shows that the upper bound on the
  worst-case response time derived from the MILP developed by Nelissen
  et al. \cite{ecrts15nelissen} can be very far from
  the actual worst-case response time in Section~\ref{sec:gap}.
\end{itemize}

\section{Proof for the Necessary Condition of Worst-Case Response Time}
\label{sec:proof-necessary}

%Suppose that $WCRT_n \leq T_n$ in this section. 
Let $\sigma$ be a fixed-priority preemptive schedule for a valid release pattern
$RP$ of the task system ${\bf T}$. We consider two cases:
\begin{itemize}
\item Case 1 when all the jobs of task $\tau_n$ in schedule $\sigma$
  have their response times no more than $T_n$: We pick an arbitrary
  job $J$ of task $\tau_n$ from $\sigma$. For this case, removing all
  the other jobs of task $\tau_n$ (except $J$) from $\sigma$ does not change the
  schedule of the remaining jobs in $\sigma$. 
\item Case 2 when there exists a job of task $\tau_n$ in schedule
  $\sigma$, in which the response time of the job is larger than
  $T_n$.  Let $J$ be the first job in the schedule $\sigma$ in which
  the response time of $J$ is strictly larger than $T_n$. For this
  case, removing all the other jobs of task $\tau_n$ (except $J$) from $\sigma$
  does not change the schedule of the remaining jobs in $\sigma$. 
\end{itemize}
In both cases, for $j=1,2,\ldots,m$, suppose that the arrival time and
finishing time of the $j$-th computation segment of job $J$ is $g_j$
and $f_j$, respectively.  In both cases, by definition, $g_1 \leq f_1
\leq g_2 \leq f_2 \leq \cdots \leq g_m \leq f_m$, and $f_m - g_1 \leq
WCRT_n$.

The following lemmas provide the necessary conditions for the
worst-case release patterns for both cases. Specifically, Condition 1
in Lemma~\ref{lemma:necessary-critical-wcrt} was also provided in
Lemma 2 by Nelissen et al. in \cite{ecrts15nelissen}. For the
completeness of this report and the correctness of the MILP, we also
include the proof of Condition 1 in
Lemma~\ref{lemma:necessary-critical-wcrt} here.

\begin{lemma}
  \label{lemma:necessary-critical-wcrt}
  If $WCRT_n \leq T_n$, then the worst-case response time of task
  $\tau_n$ happens (as necessary conditions) when
  \begin{enumerate}
  \item[Condition 1:] all the higher-priority tasks $\tau_1, \tau_2,
    \ldots, \tau_{n-1}$ only release their jobs in time intervals $[g_j,
    f_j)$ for $j=1,2,\ldots,m$, and
  \item[Condition 2:] $g_{j+1} - f_j$ is always $S_n^j$, $\forall j=1,2,\ldots,m-1$, and 
  \item[Condition 3:] all the jobs are executed with their worst-case execution times.
  \end{enumerate}
\end{lemma}
\begin{proof}
  We prove this lemma by showing that job $J$ defined at the opening
  of this section can only increase its response time by following
  these three conditions. By the assumption $WCRT_n \leq T_n$, we
  consider Case 1. 

  We start with Condition 1. Suppose that the
  schedule $\sigma$ has to execute certain higher-priority jobs in
  time interval $[t_1, g_1)$ and the processor idles right before
  $t_1$. That is, the processor does not idle between $t_1$ and
  $g_1$. In this case, we can change the arrival time of the
  computation segment $C_n^1$ of job $J$ from $g_1$ to $t_1$. This
  change in the release pattern $RP$ does not change the resulting preemptive
  schedule $\sigma$, but the response time of $J$ becomes $f_m - t_1 \geq
  f_m - g_1$ since $t_1 \leq g_1$.

  For $j=2,3,\ldots,m$, suppose that the schedule $\sigma$ has to
  execute certain higher-priority jobs to keep the processor busy in time interval $[t_j, g_j)$
  and the processor either idles or executes job $J$ right before
  $t_1$. By definition, for $j=2,3,\ldots,m$, we know that $t_j \geq
  f_{j-1}$; otherwise the $(j-1)$-th computation segment of job
  $J$ cannot be finished at time $f_{j-1}$. Similarly, we can change
  the arrival time of the computation segment $C_n^j$ of job $J$ from
  $g_j$ to $t_j$. This change in the release pattern $RP$ does not change
  the resulting preemptive schedule $\sigma$ nor the response time of $J$. 

  Let $g_j$ for $j=1,2,\ldots,m$ be the revised arrival time of the
  computation segment $C_n^j$ of job $J$ changed above. After we
  change the release pattern of job $J$, we know that the suspension
  time between the two computation segments $C_n^j$ and $C_n^{j+1}$ of
  job $J$ is exactly $g_{j+1} - f_j \leq S_n^j$ for
  $j=1,2,\ldots,m-1$. Therefore, the revised release pattern remains valid.

  Now, we can safely remove the higher jobs released before $g_1$,
  after or at $f_m$, and in time intervals $[f_1, g_2)$, $[f_2, g_3),
  \ldots, [f_{m-1}, g_m)$. Since these jobs do not (directly or
  indirectly) interfere in the execution of job $J$ at all, removing them does not
  have any impact on the execution of job $J$. Again, let $RP$ be the revised release pattern, and let $\sigma$ be its corresponding FP preemptive schedule. Now, 
  Condition 1 holds.
 
  We now revise the release pattern of the higher-priority jobs and
  $J$ for Condition 2. We start from $j=2$. If $g_j - f_{j-1} <
  S_n^{j-1}$, then we greedily perform the following steps:
  \begin{itemize}
  \item First, any higher-priority jobs released after or at $g_j$ are delayed exactly by $S_n^{j-1} - (g_j - f_{j-1})$ time units.
  \item Second, the $(j-1)$-th suspension interval of job $J$ is
    increased to suspend for exactly $S_n^{j-1}$ time units, and
    $g_\ell$ is set to $g_\ell + S_n^{j-1} - (g_j - f_{j-1})$ for
    $\ell=j,j+1,\ldots,m$.
  \end{itemize}
  With the above two steps, the schedule $\sigma$ remains almost
  unchanged by just adding $S_n^{j-1} - (g_j - f_{j-1})$ amount of
  idle time. We repeat the above procedure for $j=2,3,\ldots,m$.
  Again, let $\sigma$ be the above revised schedule with the revised
  release pattern. Now, after the adjustment, Condition 2 holds, and
  the response time of job $J$ in this schedule is larger than or
  equal to the original one.

  Condition 3 is rather trivial. If a job in schedule $\sigma$ has a
  shorter execution time, we can increase its execution time to its
  worst-case execution time. We can then adjust the release pattern
  with a similar procedure like the operations for Condition 2 to
  increase the response time of job $J$.  

  With the above discussions, we reach the conclusion of this lemma.
\end{proof}

\begin{lemma}
  \label{lemma:necessary-critical-wcrt-larger}
  If $WCRT_n > T_n$, the response time of task
  $\tau_n$ in a release pattern that satisfies Conditions 1, 2, and 3
  in Lemma~\ref{lemma:necessary-critical-wcrt} is larger than $T_n$.
\end{lemma}
\begin{proof}
  By the assumption $WCRT_n > T_n$, we consider that there exists a
  schedule $\sigma$ in which Case 2 (at the opening of the section)
  holds. The rest of the proof is identical to the proof of
  Lemma~\ref{lemma:necessary-critical-wcrt}.
\end{proof}

We now demonstrate a few properties based on
Lemma~\ref{lemma:necessary-critical-wcrt}.  By
Lemma~\ref{lemma:necessary-critical-wcrt}, for obtaining the (exact)
worst-case response time of task $\tau_n$, we simply have to examine
all the release patterns that satisfy the three conditions in
Lemma~\ref{lemma:necessary-critical-wcrt}. Specifically, Condition 1
of Lemma~\ref{lemma:necessary-critical-wcrt} implies that we can set an
\emph{offset} variable $O_{i,j}$ (with $O_{i,j} \geq 0$) to define the release time of the
first job of task $\tau_i$, arrived no earlier than $g_j$. That is,
for a given $O_{i,j}$, the first job released by task $\tau_i$ no
earlier than $g_j$ is released at time $O_{i,j} + g_j$.

With the above definitions, we can have the following properties.
These properties can be used to reduce the search space for the
worst-case response time of task $\tau_n$. The first property was also
provided in Corollary 1 in the paper by Nelissen et
al. \cite{ecrts15nelissen}.

\begin{property}
  \label{property:p1}
  For a higher-priority task $\tau_i$, there must be at least one $O_{i,j}$ equal to $0$.
\end{property}
\begin{proof}
  This is quite trivial. By Lemma~\ref{lemma:necessary-critical-wcrt},
  either task $\tau_i$ does not release any job to interfere in any
  computation segment of job $J$ or $\tau_i$ releases at least one job
  to interfere in certain computation segments of job $J$. For the
  former case, we can set $O_{i,1}$ to $0$, which does not decrease
  the resulting worst-case response time. For the latter case, if all
  $O_{i,j} > 0$ for $j=1,2,\ldots,m$, let $j^*$ be the earliest
  computation segment of job $J$ where task $\tau_i$ releases some
  jobs to interfere in. We can greedily set $O_{i,j^*}$ to $0$, which
  does not reduce the resulting worst-case response time.
\end{proof}

\begin{property}
  \label{property:p2}
  If the period of task $\tau_i$ is small enough, then the following
  property holds:
  \begin{itemize}
  \item {\bf Case when $j=1$}: If $T_i \leq S_n^1$, then $O_{i,1}$ is $0$.  
  \item {\bf Case when $j=m$}: If $T_i \leq S_n^{m-1}$, then $O_{i,m}$ is $0$.
  \item {\bf Case when $2 \leq j \leq m-1$}: If $T_i \leq S_n^{j-1}$ and $T_i
    \leq S_n^{j}$ then $O_{i,j}$ is $0$.
  \end{itemize}
  That is, task $\tau_i$ releases one job together with the $j$-th
  computation segment of the job $J$. Moreover, task $\tau_i$ also
  releases the subsequent jobs strictly periodically with period $T_i$
  until the $j$-th computation segment of job $J$ finishes.
\end{property}
\begin{proof}
  The first case is clear due to Condition 2 in
  Lemma~\ref{lemma:necessary-critical-wcrt}, i.e., the suspension time
  of the first suspension interval is exactly $S_n^1$, since the
  release pattern of task $\tau_i$ to interfere in the first
  computation segment of job $J$ is independent from the other
  computation segments. The second case is similar.

  For any $j$ with $2 \leq j \leq m-1$, the condition $T_i \leq
  S_n^{j-1}$ implies that the release pattern of task $\tau_i$ to
  interfere in the $(j-1)$-th computation segment of job $J$ is
  independent from the release pattern to interfere in the $j$-th
  computation segment. Similarly, the condition $T_i \leq S_n^{j}$
  implies that the release pattern of task $\tau_i$ to interfere in
  the $j$-th computation segment of job $J$ is independent from the
  release pattern of task $\tau_i$ to interfere in the $(j+1)$-th computation
  segment. Therefore, when $T_i \leq S_n^{j-1}$ and $T_i \leq
  S_n^{j}$, the release pattern of task $\tau_i$ to interfere in the
  $j$-th computation segment of job $J$ is independent from the other
  computation segments.

  Moreover, when the release pattern of task $\tau_i$ to interfere in
  the $j$-th computation segment of job $J$ is independent from the
  other computation segments, the worst-case release pattern of task
  $\tau_i$ to interfere in the $j$-th computation segment of job $J$
  is to release 1) one job together with the $j$-th computation
  segment of the job $J$, and 2) the subsequent jobs strictly
  periodically with period $T_i$ until the $j$-th computation segment
  of job $J$ finishes.
\end{proof}

\begin{property}
  \label{property:p3}
  If $T_i \geq T_n$ and $WCRT_n \leq T_n$, then a higher-priority task $\tau_i$ only
  releases one job together with one of the $m$ computation segments
  of the job (under analysis) of task $\tau_n$.
\end{property}
\begin{proof}
  This comes from Condition 1 in
  Lemma~\ref{lemma:necessary-critical-wcrt} and Property
  \ref{property:p1}.
\end{proof}

\begin{property}
 \label{property:p4}
 If $T_i - C_i$ is small enough, then the following property of task
 $\tau_i$ holds:
  \begin{itemize}
  \item {\bf Case when $j=1$}: If $T_i - C_i \leq S_n^1$, then $O_{i,1}$ is $0$.  
  \item {\bf Case when $j=m$}: If $T_i - C_i \leq S_n^{m-1}$, then $O_{i,m}$ is $0$.
  \item {\bf Case when $2 \leq j \leq m-1$}: If $T_i -C_i \leq
    S_n^{j-1}$ and $T_i - C_i
    \leq S_n^{j}$ then $O_{i,j}$ is $0$.
  \end{itemize}
  That is, task $\tau_i$ releases one job together with the $j$-th
  computation segment of the job $J$. Moreover, task $\tau_i$ also
  releases the subsequent jobs strictly periodically with period $T_i$
  until the $j$-th computation segment of job $J$ finishes.  
\end{property}
\begin{proof}
  This is a simple extension of Property~\ref{property:p2}.
\end{proof}

\section{Computational Complexity}
\label{sec:computational-complexity-fp}

In this section, we will prove that the schedulability test problem
for FP preemptive scheduling even with only one
\emph{segmented} self-suspending task as the lowest-priority task in
the task system is $co{\cal NP}$-hard in the strong
sense. Specifically, we will also show that our reduction implies
that finding whether there exists a feasible priority assignment under
FP scheduling for constrained-deadline task systems is
also $co{\cal NP}$-hard in the strong sense. We will
first consider constrained-deadline task systems and then revise the
reduction to consider implicit-deadline task systems.

Our reduction is from the 3-PARTITION problem
\cite{b:Gary79}:\footnote{For notational consistency and brevity in
  our reduction, we index the $3M$ integer numbers from $2$.}
\begin{definition}[3-PARTITION Problem] We are given a positive
  integer $V$, a positive integer $M$, and a set of $3M$ integer
  numbers $\setof{v_2, v_3, \ldots, v_{3M+1}}$ with the condition
  $\sum_{i=2}^{3M+1} v_i = MV$, in which $1 \leq V/4 < v_i < V/2$ and $M
  \geq 3$.  Therefore, $V \geq 3$. \\{\bf
    Objective:} The problem is to partition the given $3M$ integer
  numbers into $M$ disjoint sets ${\bf V}_1, {\bf V}_2, \ldots, {\bf
    V}_M$ such that the sum of the numbers in each set ${\bf V}_i$ for
  $i=1,2,\ldots,M$ is $V$, i.e., $\sum_{v_j \in {\bf V}_i} v_j = V$. \hfill\myendproof
\end{definition}
The decision version of the 3-PARTITION problem to verify whether such
a partition into $M$ disjoint sets exists or not is known ${\cal
  NP}$-complete in the strong sense \cite{b:Gary79} when $M \geq 3$.

\subsection{Constrained-Deadline Task Systems}
\label{sec:conp-hard-constrained}

\begin{definition}[Reduction to a constrained-deadline system]
\label{def:reduced-taskset-constrained}
For a given input instance of the 3-PARTITION problem, we construct
$n=3MV+2$ sporadic tasks as follows:
\begin{compactitem}
\item For task $\tau_1$, we set $C_1 = V, S_1 = 0, D_1=V, T_1=3V$.
\item For task $\tau_i$ with $i=2,3,\ldots, 3M+1$, we set $C_i = v_i,
  S_i = 0, T_i=21MV$ and $D_i=3MV/2$ if $M$ is an even number or
  $D_i=3MV/2+ V/2$ if $M$ is an odd number.
\item For task $\tau_{3M+2}$, we create a segmented self-suspending
  task with $M$ computation segments separated by $M-1$
  self-suspension intervals\footnote{The first version of the proof
    uses $6V$ for $S_{3M+2}^j$ for $j=1,2,\ldots,M-1$. By applying
    Property~\ref{property:p3}, we can also set $S_{3M+2}^j$ to $2V$
    and $D_{3M+2} = M (4V + 1) -V + 2V(M-1)= 6MV+M-3V$.}, i.e., $m=m_{3M+2}=M$, in which $C_{3M+2}^j= V+1$ for
  $j=1,2,\ldots,M$, $S_{3M+2}^j = 6V$ for $j=1,2,\ldots,M-1$,
  $D_{3M+2} = M (4V + 1) -V + 6V(M-1)= 10MV+M-7V$, and $T_{3M+2}=21MV$.
\end{compactitem}
Due to the stringent relative deadline of task $\tau_1$, it
must be assigned as the highest-priority task. Moreover, the $3M$
tasks, i.e., $\tau_2, \tau_3, \ldots, \tau_{3M+1}$, created by using
the integer numbers from the 3-PARTITION problem instance are assigned
lower priorities than task $\tau_1$ and higher priorities than task
$\tau_{3M+2}$. Since the integer numbers in the 3-PARTITION problem
instance are given in an arbitrary order, without loss of generality,
we index the tasks in $\tau_2, \tau_3, \ldots, \tau_{3M+1}$ by the
given priority assignment, i.e., a lower-indexed task has higher
priority. (In fact, we can also assign all these $3M$ tasks with the
same priority level.)  \hfill\myendproof
\end{definition}

For the rest of the proof, the task set created in
Definition~\ref{def:reduced-taskset-constrained} is referred to as
${\bf T}^{red}$.

\begin{lemma}
  \label{lemma:feasibility-sporadic-tasks}
  Tasks $\tau_1, \tau_2, \ldots, \tau_{3M+1}$ in ${\bf T}^{red}$ can
  meet their deadlines under the specified FP
  scheduling.
\end{lemma}
\begin{proof}
  In FP scheduling, the segmented
  self-suspending task $\tau_{3M+2}$ in ${\bf T}^{red}$ has no impact
  on the schedule of the higher-priority tasks. Therefore, we can use
  the standard schedulability test for FP
  scheduling to verify their schedulability. The schedulability of
  task $\tau_1$ is obvious since $C_1 \leq D_1$. For
  $i=2,3,\ldots,3M+1$, task $\tau_i$ is schedulable under
  FP scheduling since $C_i + \sum_{j=1}^{i-1}
  \ceiling{\frac{D_i}{T_j}}C_j =
  \ceiling{\frac{D_i}{T_1}}C_1+\sum_{j=2}^{i} C_j \leq
  \ceiling{\frac{D_i}{3V}}V+MV = D_i$, where the last equality is due
  to $\ceiling{\frac{D_i}{3V}} = \ceiling{\frac{3MV/2}{3V}} = M/2$
  when $M$ is an even number and $\ceiling{\frac{D_i}{3V}} =
  \ceiling{\frac{3MV/2+V/2}{3V}} = (M+1)/2$ when $M$ is an add number.
\end{proof}

% To analyze whether task $\tau_{3M+2}$ in ${\bf T}^{red}$ can be
% schedulable under FP scheduling, we can adopt the exact worst-case
% response time analysis by Nelissen et al. \cite[Section
% IV]{ecrts15nelissen}. However, as the MILP constructed by Nelissen et
% al. \cite{ecrts15nelissen} has some implicit properties that were
% used without a proof, to make our proof complete, the author has also
% filed a comprehensive report \cite{DBLP:journals/corr/abs-1604.xxxx}
% to support the correctness of the MILP developed in
% \cite{ecrts15nelissen}. 
The worst-case response time of task $\tau_{3M+2}$
happens by using one of the release patterns with the conditions in
Lemma~\ref{lemma:wcrt-pattern-constrained}:
\begin{lemma}
  \label{lemma:wcrt-pattern-constrained}
  The worst-case response time of task $\tau_{3M+2}$ in ${\bf
    T}^{red}$ under FP scheduling happens under
  the following necessary conditions:
  \begin{enumerate}
  \item Task $\tau_{3M+2}$ releases a job at time $0$. This job
    requests the worst-case execution time per computation segment and
    suspends in each self-suspension interval exactly equal to its
    worst case.
  \item Task $\tau_1$ always releases one job together with each
    computation segment of the job (released at time $0$) of task
    $\tau_{3M+2}$, and releases the subsequent jobs strictly
    periodically with period $3V$ until a computation segment of task
    $\tau_{3M+2}$ finishes. Task $\tau_1$ never releases any job when
    task $\tau_{3M+2}$ suspends itself.
  \item For $i=2,3,\ldots,3M+1$, task $\tau_i$ only releases one job
    together with one of the $M$ computation segments of the job
    (released at time $0$) of task $\tau_{3M+2}$.
  \end{enumerate}
  All the jobs and all the computation segments are executed with their
  worst-case execution time specifications.
\end{lemma}
\begin{proof}
% There are a few properties for the necessary condition of the
% worst-case release pattern:
% \begin{enumerate}
% \item We only have to check one job execution of task $\tau_n$.
% \item All the higher-priority tasks only release
%   their jobs between the arrival time and the finishing time of each
%   computation segment of task $\tau_n$. 
% \item All the self-suspension intervals of task
%   $\tau_n$ always take the worst case.  All the jobs and all the
%   computation segments are executed with their worst-case execution
%   time specifications.
% % \item It also implicitly assumes that $N_{i,j}$ should be set to
% %   $\ceiling{\frac{R_j - O_{i,j}}{T_i}}$ or $\ceiling{\frac{R_j -
% %       O_{i,j}}{T_i}}-1$ if $\ceiling{\frac{R_j - O_{i,j}}{T_i}}$ is
% %   positive in Eq.~\eqref{eq:N-bounds} and Eq.~\eqref{eq:O-distance}.
% \end{enumerate}
% The proof of the above properties can be found in Lemma~\ref{lemma:necessary-critical-wcrt}.
% We can now use the above properties to prove Lemma~\ref{lemma:wcrt-pattern-constrained}. 
For task $\tau_1$ in ${\bf T}^{red}$, since $T_1 < S_n^j$ for any $j=1,2,\ldots,M-1$, the
release pattern of task $\tau_1$ is independent from the computation
segments. This is formally proved in Property~\ref{property:p2} in Section~\ref{sec:proof-necessary}.
Moreover, since $T_i > D_n$
for $i=2,3,\ldots,3M+1$, such a higher-priority task $\tau_i$ in ${\bf T}^{red}$ only
releases one job together with one of the $M$ computation segments of
the job (under analysis) of task $\tau_n$. This is formally proved in
Property~\ref{property:p3} in Section~\ref{sec:proof-necessary}.
By putting all the above conditions together, we reach the conclusion for
Lemma~\ref{lemma:wcrt-pattern-constrained}. 
\end{proof}

For the $j$-th computation segment of task $\tau_{3M+2}$, suppose that
${\bf T}_j \subseteq \setof{\tau_2, \tau_3, \ldots, \tau_{3M+1}}$ is
the set of the tasks released together with $C_{3M+2}^j$ (under the
third condition in Lemma~\ref{lemma:wcrt-pattern-constrained}).  For
notational brevity, let $w_j$ be $\sum_{\tau_i \in {\bf T}_j} C_i$. By
definition, $w_j$ is a non-negative integer.  Together with the second
condition in Lemma~\ref{lemma:wcrt-pattern-constrained}, we can use
the standard time demand analysis to analyze the worst-case response
time $R_j$ of the $j$-th computation segment of task $\tau_{3M+2}$
(after it is released) under the higher-priority interference due to
$\setof{\tau_1} \cup {\bf T}_j$.
The response time $R_j$ of a computation segment $C_{3M+2}^j$ is defined as
the finishing time of the computation segment minus the arrival time
of the computation segment.

For a given task set ${\bf T}_j$ (i.e., a given non-negative integer $w_j$),
$R_j$ is the minimum $t$ with  $t > 0$  such that
\[
C_{3M+2}^j + (\sum_{\tau_i \in {\bf T}_j} C_i) +\ceiling{\frac{t}{T_1}}C_1=  V + 1 + w_j + \ceiling{\frac{t}{3V}}V = t.
\]
Since $R_j$ only depends on
the non-negative integer $w_j$, we use $R(w_j)$ to represent $R_j$ for
a given ${\bf T}_j$.
We know that $V
+ 1 + w_j + \ceiling{\frac{t}{3V}}V = t$ happens with $\ell\cdot 3V < t
\leq (\ell+1)\cdot 3V$ for a certain non-negative integer $\ell$.  That is,
$V + 1 + w_j + \ceiling{\frac{t}{3V}}V > t$ when $t$ is $\ell\cdot 3V$
and $V + 1 + w_j + \ceiling{\frac{t}{3V}}V \leq t$ when $t$ is
$(\ell+1)3V$. We know that $\ell$ is
$\ceiling{\frac{V+1+w_j}{2V}}-1$. Moreover,
\begin{align*}
 R(w_j) &=  \ell\cdot 3V + V+(V+1+w_j - \ell\cdot 2V) \\
  &= 2V+1+w_j + \ell\cdot V \\
  &= V+1+w_j + \ceiling{\frac{V+1+w_j}{2V}} V.
\end{align*}
This leads to three cases that are of interest:
\begin{equation}
  \label{eq:response-time-np-hard}
R(w_j) =
\begin{cases}
  2V+1+w_j & \mbox{ if }w_j \leq V -1\\ 
  4V+1 & \mbox{ if } w_j = V \\ 
  V+1+w_j + \ceiling{\frac{V+1+w_j}{2V}} V& \mbox{ if } w_j > V
\end{cases}
\end{equation}
For example, if $w_j = 3V-1$, then $R(w_j)$ is $6V$; if $w_j$ is $3V$,
then $R(w_j)$ is $7V+1$. 

With the above discussions, we can now conclude that the unique
condition when task $\tau_{3M+2}$ misses its deadline in the following
lemma.

\begin{lemma}
  \label{lemma:schedulability-equivalent}
  Suppose that ${\bf T}_j \subseteq \setof{\tau_2, \tau_3, \ldots,
    \tau_{3M+1}}$ and ${\bf T}_i \cap {\bf T}_j = \emptyset$ when $i
  \neq j$. Let $w_j = \sum_{\tau_i \in {\bf T}_j} C_i$.  If a task
  partition ${\bf T}_1, {\bf T}_2, \ldots, {\bf T}_M$ exists such that
  $\sum_{j=1}^{M} R(w_j) > M(4V+1)-V$ with $R(w_j)$ defined in
  Eq.~\eqref{eq:response-time-np-hard}, then task $\tau_{3M+2}$ misses
  its deadline in the worst case; otherwise, task $\tau_{3M+2}$ always
  meets its deadline.
\end{lemma}
\begin{proof}
  By Lemma~\ref{lemma:wcrt-pattern-constrained}, task $\tau_{3M+2}$ in
  ${\bf T}^{red}$ is not schedulable under the fixed-priority
  preemptive scheduling if and only if there exists a task partition ${\bf T}_1,
  {\bf T}_2, \ldots, {\bf T}_M$ such that $\sum_{j=1}^{M} R(w_j) +
  \sum_{j=1}^{M-1} S_{3M+2}^j = 6(M-1)V + \sum_{j=1}^{M} R(w_j) >
  D_{3M+2} = M (4V + 1) + 6V(M-1) -V$. This concludes the proof.
\end{proof}

Instead of investigating the combinations of the task partitions, we
analyze the corresponding total worst-case response time
$\sum_{j=1}^{M} R(w_j)$ for the $M$ computation segments of task
$\tau_{3M+2}$ (by excluding the self-suspension time) by considering
different \emph{non-negative integer assignments} $w_1, w_2, \ldots,
w_{M}$ with $\sum_{i=1}^{M} w_i = MV$ and $w_i \geq 0$ in the following
lemmas.

\begin{lemma}
  \label{lemma:wcrt-all-the-same}
  If $w_1 = w_2 = \cdots = w_{M} = V$, then 
\[
  \sum_{j=1}^{M} R(w_j) = M(4V+1),
\]
where $R(w_j)$ is defined in Eq.~\eqref{eq:response-time-np-hard}.
\end{lemma}
\begin{proof}
  This comes directly by Eq.~\eqref{eq:response-time-np-hard}.
\end{proof}

\begin{lemma}
  \label{lemma:wcrt-one-different}
  For any non-negative integer assignment for $w_1, w_2, \ldots, w_M$
  with $\sum_{i=1}^{M} w_i = MV$, if there exists a certain index $j$
  with $w_j \neq V$, then
\[
  \sum_{j=1}^{M} R(w_j) \leq  M(4V+1)-V,
\]
where $R(w_j)$ is defined in Eq.~\eqref{eq:response-time-np-hard}.
\end{lemma}
\begin{proof}
  Let ${\bf X}$ be the set of indexes such that $0 \leq w_j < V$ for
  any $j \in {\bf X}$. Similarly, let ${\bf Y}$ be the set of indexes
  such that $V < w_j$ for any $j \in {\bf Y}$.  If $j \notin {\bf
    X}\cup{\bf Y}$, then $w_j$ is $V$.

  If there exists $j$ in ${\bf Y}$ with $w_j > 2V$, since
  $\sum_{i=1}^{M} w_i = MV$, there must exist an index $i$ in ${\bf
    X}$ with $w_i < V$. We can increase $w_i$ to $w_i'=V$, which increases
  the worst-case response time $R(w_i)$ by $2V-w_i$ (i.e., from
  $2V+1+w_i$ to $4V+1$). Simultaneously, we
  reduce $w_j$ to $w_j' = w_j - (V-w_i) > V$. Therefore,
  $w_i+w_j=w_i'+w_j'$. Moreover, the reduction of $w_j$ to $w_j'$ also reduces the
  worst-case response time $R(w_j)$ by case 1) $V-w_i$ if
  $\ceiling{\frac{V+1+w_j'}{2V}}$ is equal to $\ceiling{\frac{V+1+w_j}{2V}}$, and by case 2) $V-w_i+V$ if $\ceiling{\frac{V+1+w_j'}{2V}}$ is not equal to $\ceiling{\frac{V+1+w_j}{2V}}$. In both
  cases, we can easily see that the worst-case response time is not
  decreased in the new integer assignment.
 Moreover, the index $j$
  remains in ${\bf Y}$ and the index $i$ is removed from set ${\bf
    X}$. We repeat the above step until all the indexes $j$ in ${\bf
    Y}$ are with $w_j \leq 2V$.

  It is clear that ${\bf X}$ and ${\bf Y}$ are both non-empty after the above step.
  For the rest of the proof, let ${\bf X}$ and ${\bf Y}$ be defined after finishing the above step. 
  Therefore, the condition 
  $w_j \leq 2V$ holds for any $j \in {\bf Y}$. 
 Due to the pigeon-hole
  principle, when ${\bf Y}$ is not an empty set, ${\bf X}$ is also not
  an empty set.  Moreover, for an element $i$ in ${\bf X}$, there must
  be a subset ${\bf Y}'\subseteq {\bf Y}$ and an index $\ell \in {\bf
    Y}'$ such that
  \[ \sum_{j \in {\bf Y}'}
  (w_j-V) \geq V-w_i > \sum_{j \in {\bf Y}'\setminus\setof{\ell}}  (w_j-V).
  \]
  That is, we want to adjust $w_i$ to $V$ (i.e., $w_i$ is increased by
  $V-w_i$), and the set ${\bf Y}'\setminus\setof{\ell}$ is not
  enough to match the integer adjustment $V-w_i$ and the set ${\bf
    Y}'$ is enough to match the integer adjustment $V-w_i$.  We now
  increase $w_i$ to $V$, which increases the worst-case response time
  $R(w_i)$ by $2V-w_i$.  Simultaneously, we reduce $w_j$ to $V$ for
  every $j \in {\bf Y}' \setminus \setof{\ell}$ and reduce $w_\ell$ to
  $w_\ell'=w_\ell - (V-w_i-\sum_{j \in {\bf Y}'\setminus\setof{\ell}}
  (w_j-V))$. Since $V < w_j \leq 2V$ for any $j \in {\bf Y}'$ before
  the adjustment, the adjustment reduces the worst-case response time
  $R(w_j)$ by $w_j - V$ if $j\neq \ell$ and reduces $R(w_\ell)$ by
  $w_\ell - w'_\ell$.  Therefore, the adjustment reduces $\sum_{j \in
    {\bf Y}'} R(w_j)$ by exactly $V-w_i$. Therefore, the adjustment in
  this step to change $w_i$ in ${\bf X}$ and $w_j$ in ${\bf Y}'$ increases the
  overall worst-case response time by exactly $V$ time units.

  By adjusting with the above procedure repeatedly, we will reach the
  integer assignment $w_1 = w_2 = \cdots = w_{M} = V$ with bounded
  increase of the worst-case response time.  As a result, we can
  conclude that $\sum_{j=1}^{M} R(w_j) \leq M(4V+1)-|{\bf X}|V$. By the
  assumption $\sum_{i=1}^{M} w_i = MV$ and the existence of $w_j \neq
  V$ for some $j$, we know that $|{\bf X}|$ must be at least
  $1$. Therefore, we reach the conclusion.
\end{proof}

We use an example to illustrate how the procedure in Lemma
  \ref{lemma:wcrt-one-different} operates. Suppose that $w_1=0, w_2 =
  3.5V, w_3=0.4V, w_4=0.6V, w_5=1.5V, w_6=0$ when $M=6$ and $V$ is an
  integer multiple of $10$. We will start from ${\bf
    X}=\setof{1,3,4,6}$ and ${\bf Y} = \setof{2,5}$. As shown in
  Table~\ref{tab:example-lemma-merge}, the operation makes
  $\sum_{j=1}^{6}R(w_j)$ increase. Note that the conclusion
  $\sum_{j=1}^{M} R(w_j) \leq M(4V+1)-|{\bf X}|V$ in
  Lemma~\ref{lemma:wcrt-one-different} was for $|{\bf X}|=\setof{3,4}$
  in this example when $w_j < 2V$ for any $j \in {\bf Y}$.

  \begin{table}[t]
    \centering\scalebox{0.7}{
    \begin{tabular}{|c|c|c|c|c|c|c|c|c|c|c|c||c|c||c|c||c|}
      \hline
        $w_1$ & $R(w_1)$
      &   $w_2$ & $R(w_2)$
      &   $w_3$ & $R(w_3)$
      &   $w_4$ & $R(w_4)$
      &   $w_5$ & $R(w_5)$
      &   $w_6$ & $R(w_6)$
      &   ${\bf X}$ & ${\bf Y}$ & $\sum_{j=1}^{6}R(w_j)$\\
\hline
0 & $2V+1$ & $3.5V$ & $7.5V+1$ & $0.4V$ & $2.4V+1$ & $0.6V$ & $2.6V+1$
& $1.5V$ & $4.5V+1$ & $0$ & $2V+1$ & $\setof{1,3,4,6}$ &
$\setof{2,5}$ & $21V+6$\\
\hline
$V$ & $4V+1$ & $2.5V$ & $5.5V+1$ & \ditto& \ditto& \ditto& \ditto 
& \ditto & \ditto & \ditto & \ditto & $\setof{3,4,6}$ & $\setof{2,5}$
& $21V+6$\\
\hline
\ditto & \ditto & $1.5V$ & $4.5V+1$ & \ditto& \ditto& \ditto& \ditto 
& \ditto & \ditto & $V$ & $4V+1$ & $\setof{3,4}$ & $\setof{2,5}$
& $22V+6$\\
\hline
\ditto & \ditto & $V$ & $4V+1$ & $V$ & $4V+1$& \ditto& \ditto 
& $1.4V$ & $4.4V+1$& \ditto & \ditto & $\setof{4}$ & $\setof{5}$
& $23V+6$\\
\hline
\ditto & \ditto & \ditto & \ditto & \ditto & \ditto & $V$ & $4V+1$
& $1V$ & $4V+1$& \ditto & \ditto & $\emptyset$ & $\emptyset$
& $24V+6$\\
\hline
    \end{tabular}}
    \caption{An example of Lemma~\ref{lemma:wcrt-one-different}}
    \label{tab:example-lemma-merge}
  \end{table}

We can now conclude the $co{\cal NP}$-hardness.

\begin{theorem}
\label{thm:conp-hard-constrained}
The schedulability analysis for FP scheduling
even with only one segmented self-suspending task as the
lowest-priority task in the sporadic task system is $co{\cal NP}$-hard
in the strong sense, when the number of self-suspending intervals in
the self-suspending task is larger than or equal to $2$ and $D_i \leq
T_i$ for every task $\tau_i$.
\end{theorem}
\begin{proof}
  The reduction in Definition~\ref{def:reduced-taskset-constrained}
  requires polynomial time. Moreover, by
  Lemmas~\ref{lemma:wcrt-pattern-constrained},~\ref{lemma:schedulability-equivalent},~\ref{lemma:wcrt-all-the-same},~and~\ref{lemma:wcrt-one-different},
  a feasible solution of the 3-PARTITION problem for the input
  instance exists if and only if task $\tau_{3M+2}$ is not schedulable
  by the FP scheduling when $M \geq 3$. Therefore, this concludes the proof.
\end{proof}

\begin{corollary}
  \label{cor:conp-hard-constrained-scheduler-design}
  Validating whether there exists a feasible priority assignment is
  $co{\cal NP}$-hard in the strong sense for constrained-deadline
  segmented self-suspending task systems.
\end{corollary}
\begin{proof}
  This comes directly from Theorem~\ref{thm:conp-hard-constrained} and
  the only possible priority level for task $\tau_{3M+2}$ to be
  feasible in ${\bf T}^{red}$.
\end{proof}

\subsection{Implicit-Deadline Task Systems}

The $co{\cal NP}$-hardness in the strong sense for testing the
schedulability of task $\tau_n$ under FP scheduling can be easily
proved with the same input as in ${\bf T}^{red}$ by changing the
periods of the tasks as follows:
\begin{compactitem}
\item For task $\tau_1$, we set $D_1=3V, T_1=3V$.
\item For task $\tau_i$ with $i=2,..., 3M+1$, we set $T_i=D_i=10MV+M-7V$.
\item For task $\tau_{3M+2}$, we set $T_{3M+2} = D_{3M+2} =
  10MV+M-7V$.
\end{compactitem}
Assume that $\tau_{3M+2}$ is the lowest-priority task. It is not
difficult to see that all the conditions in
Lemma~\ref{lemma:wcrt-pattern-constrained} still hold for testing whether task $\tau_{3M+2}$ can meet its deadline or not (but not for the worst-case response time if task $\tau_{3M+2}$ misses the deadline). Therefore, the
schedulability analysis for FP scheduling even with only one segmented
self-suspending task as the lowest-priority task in the sporadic task
system is $co{\cal NP}$-hard in the strong sense, when the number of
self-suspending intervals in the self-suspending task is larger than or
equal to $2$ and $D_i = T_i$ for every task $\tau_i$.

However, the above argument does not hold if we assign task
$\tau_{3M+2}$ to the highest-priority level. Therefore, the above
proof does not support a similar conclusion for implicit-deadlien task
systems to that for constrained-deadline task systems in
Corollary~\ref{cor:conp-hard-constrained-scheduler-design}.

\section{MILP Approaches} 
\label{sec:proof-milp}

Even though the properties in
Lemma~\ref{lemma:necessary-critical-wcrt} provide the necessary
conditions for the worst-case response time, finding the worst-case release
pattern is in fact a hard problem as shown in the analysis in
Section~\ref{sec:computational-complexity-fp}. However, if we can
tolerate exponential time complexity, is there a strategy that can
find the worst-case pattern based on
Lemma~\ref{lemma:necessary-critical-wcrt} safely without performing
exhaustive searches? One possibility is to model the problem as an MILP,
which has been already presented by Nelissen et al. \cite{ecrts15nelissen}.

The worst-case response time analysis by Nelissen et
al. \cite{ecrts15nelissen} is based on the following mixed-integer
linear programming (MILP):
\begin{subequations}
  \label{eq:MILP}\footnotesize{
  \begin{align}
    \mbox{\bf maximize: } \;\;\;\;S_n + \sum_{j=1}^{m} R_j \label{eq:objective}\\
    \mbox{\bf subject to:} \qquad\qquad\qquad&\nonumber\\
     R_j = C_n^j + \sum_{i=1}^{n-1} N_{i,j} \times C_i, \qquad&
     \forall j=1,\ldots,m \label{eq:R_j}\\
     O_{i,j} \geq 0, \qquad& \forall i=1,\ldots,n-1, \forall
     j=1,\ldots,m \label{eq:O-positive}\\
     O_{i,j+1} \geq O_{i,j} + N_{i,j}\times T_i - (R_j + S_n^j), \qquad \label{eq:O-distance}&
    \forall i=1,\ldots,n-1, \forall j=1,\ldots,m-1\\
     0 \leq N_{i,j}  \leq \ceiling{\frac{R_j - O_{i,j}}{T_i}}, \qquad&
    \forall i=1,\ldots,n-1, \forall j=1,\ldots,m \label{eq:N-bounds}\\
     N_{i,j} \mbox{ is an integer }, \qquad&
    \forall i=1,\ldots,n-1, \forall j=1,\ldots,m \label{eq:N-integer}\\
  R_j \leq UB_{ss,j} \qquad &\forall j=1,2,\ldots,m, \label{eq:total-segment-upperbound} \\
  S_n + \sum_{j=1}^{m} R_j \leq UB_{ss} \label{eq:total-upperbound}\\
    \mbox{ Eq. \eqref{eq:Rj>rel-original} holds}.
  \end{align}}
\end{subequations}
In the above MILP, the objective function $S_n + \sum_{j=1}^{m} R_j$
is the worst-case response time of task $\tau_n$, where $R_j$ is a variable
(as a real number) that represents the response time of
the $j$-th computation segment $C_n^j$ of task $\tau_n$. The variable
$O_{i,j}$ defines the \emph{offset} of the first job of a
higher-priority task $\tau_i$ released no earlier than the arrival
time of the $j$-th computation segment $C_n^j$ of task $\tau_n$. That
is, if the arrival time of $C_n^j$ is $t_j$, then the first job of
task $\tau_i$ released at or after $t_j$ is at time $t_j +
O_{i,j}$. The integer variable $N_{i,j}$ defines the maximum number of
jobs of a higher-priority task $\tau_i$ that are released to
\emph{successfully interfere} in the computation segment $C_n^j$ of task
$\tau_n$.

The three additional constraints, expressed by Eq. (9), Eq. (11), and
Eq. (16), in the MILP in \cite{ecrts15nelissen} are expressed here by
Eq.~\eqref{eq:total-segment-upperbound}, Eq.~\eqref{eq:total-upperbound}, and
Eq.~\eqref{eq:Rj>rel-original}, respectively. Here, $UB_{ss}$ is defined as the
upper bound on the worst-case response time of task $\tau_n$, and
$UB_{ss,j}$ is defined as the upper bound on the worst-case response
time of the $j$-th computation segment of task $\tau_n$.  Later in
this section, we will show that the condition in Eq.~\eqref{eq:R_j}
may over-estimate the worst-case response time. Therefore, the
additional constraint (expressed by Eq. (16), in the MILP in
\cite{ecrts15nelissen}) is used to reduce the pessimism as follows: {\small
  \begin{equation}
    \label{eq:Rj>rel-original}
\forall i=1,2,\ldots,n-1, j=1,2,\ldots,m, \qquad R_j > rel_{i,j} + \sum_{\ell=1}^{n-1}\max\left\{0,
     \floor{\frac{O_{\ell,j}+N_{\ell,j}T_\ell - rel_{i,j}}{T_\ell}}C_\ell\right\},     
 \end{equation}}where $rel_{i,j} = O_{i,j} + (N_{i,j}-1) T_i$. This means that
the (total) execution time of all the higher-priority jobs (by tasks $\tau_1, \tau_2, \ldots, \tau_{n-1}$) released
after $rel_{i,j}$ should be less than $R_j-rel_{i,j}$.

% We will discuss the usefulness of these additional constraints
% Eq.~\eqref{eq:total-segment-upperbound}, Eq.~\eqref{eq:total-upperbound}, and
% Eq.~\eqref{eq:Rj>rel-original} in Section~\ref{sec:gap}. 

Here, we first explain why the MILP by utilizing only the constraints
from Eq.~\eqref{eq:R_j} to Eq.~\eqref{eq:N-integer} is already a safe
(but \emph{not tight/exact}) result based on
Lemma~\ref{lemma:necessary-critical-wcrt}. Therefore, this also leads to
the motivation to examine the pessimism by different combinations of
the additional constraints Eq.~\eqref{eq:total-segment-upperbound},
Eq.~\eqref{eq:total-upperbound}, and Eq.~\eqref{eq:Rj>rel-original} in
Section~\ref{sec:gap}.

\subsection{MILP by Using Lemma~\ref{lemma:necessary-critical-wcrt}}

We only consider the release patterns of the tasks in ${\bf T}$, where
all the three conditions in Lemma~\ref{lemma:necessary-critical-wcrt}
hold.  Let $r_{i,j}$ be the arrival time of the first job of task
$\tau_i$ arrived after or at time $g_j$ in a concrete release pattern,
in which all the three conditions in
Lemma~\ref{lemma:necessary-critical-wcrt} hold. If task $\tau_i$ does
not release any job after or at time $g_j$, we set $r_{i,j}$ to
$\infty$.\footnote{With the discussions below, we will later set
  $r_{i,j}$ to $f_j + T_j$ for such a case (but not release any job of
  task $\tau_i$ at time $f_j+T_j$).} By the minimum inter-arrival time
constraint of task $\tau_i$, we know that task $\tau_i$ cannot release
any job in time interval $(r_{i,j+1}-T_i, r_{i,j+1})$. That is, in
this release pattern, there are \emph{at most}
$\floor{\frac{r_{i,j+1}-T_i - r_{i,j}}{T_i}}+1 \leq \frac{r_{i,j+1} -
  r_{i,j}}{T_i}$ jobs from task $\tau_i$ that can interfere in the
$j$-th computation segment of job $J$.

Let $N_{i,j}$ be the number of jobs of a higher-priority task $\tau_i$
released in time interval $[g_j, f_j)$ in this release pattern. By
definition, $N_{i,j}$ is a non-negative integer.  The maximum number
of jobs that task $\tau_i$ can release in time interval $[r_{i,j},
f_j)$ in this release pattern can be expressed by the following inequality:
\begin{equation}
0 \leq N_{i,j} \leq \max\left\{0, \ceiling{\frac{f_j-r_{i,j}}{T_i}}\right\}, \qquad \forall i=1,\ldots,n-1, j=1,\ldots,m.   \label{eq:N-ij-1-pre}  
\end{equation}
The reason to put $\max\left\{0,
  \ceiling{\frac{f_j-r_{i,j}}{T_i}}\right\}$ instead of only
$\ceiling{\frac{f_j-r_{i,j}}{T_i}}$ in the right-hand side of
Eq.~\eqref{eq:N-ij-1-pre} is to avoid the case that
$\ceiling{\frac{f_j-r_{i,j}}{T_i}} < 0$, which is possible if $r_{i,j}
> f_j + T_i$. 

There is one simple trick regarding to the setting of $r_{i,j}$. If
$r_{i,j} > f_j + T_i$, for this release pattern, we know that 1) task
$\tau_i$ does not release any job to interfere in the $j$-th
computation segment of job $J$ and 2) the number of jobs of task
$\tau_i$ that are released to interfere in the $(j-1)$-th computation segment
of job $J$ is purely dominated by $\max\left\{0,
  \ceiling{\frac{f_{j-1}-r_{i,j-1}}{T_i}}\right\}$.  Therefore, if
$r_{i,j} > f_j + T_i$, we can safely set $r_{i,j}$ to $f_j+T_i$
(but we do not change the release pattern to release a job of task
$\tau_i$ at time $f_j + T_i$ for such a case). With this, we can then rephrase
Eq.~\eqref{eq:N-ij-1-pre} into
\begin{equation}
0 \leq N_{i,j} \leq \ceiling{\frac{f_j-r_{i,j}}{T_i}}, \qquad \forall i=1,\ldots,n-1, j=1,\ldots,m.   \label{eq:N-ij-1}  
\end{equation}

By earlier discussions, we also have
\begin{equation}
  N_{i,j} \leq \frac{r_{i,j+1} - r_{i,j}}{T_i}, \qquad \forall i=1,\ldots,n-1, j=1,\ldots,m-1.   \label{eq:N-ij-2}  
\end{equation}

By Condition 1 and Condition 3 in Lemma~\ref{lemma:necessary-critical-wcrt}, we also know that
\begin{equation}
  \label{eq:f_j}
  f_j \leq g_j + C_n^j + \sum_{i=1}^{n-1} N_{i,j} \times C_i \qquad \forall j=1,2,\ldots, m.
\end{equation}

Without loss of generality, we can set $g_1$ to $0$. By Condition 2 in
Lemma~\ref{lemma:necessary-critical-wcrt}, we have
\begin{equation}
  \label{eq:g_j}
g_1 = 0  \mbox{ and }  g_j = f_{j-1} + S_n^{j-1} \qquad \forall j=2,3,\ldots, m.
\end{equation}

Now we can conclude the following theorem:

\begin{theorem}
  \label{thm:MILP-v1}
  Suppose that $g_j, f_j, r_{i,j}$ are variables of real numbers and
  $N_{i,j}$ are variables for non-negative integer numbers for
  $i=1,2,\ldots,n-1$ and for $j=1,2,\ldots,m$. The optimal solution of
  the following MILP is a safe upper bound on the worst-case response time of task
  $\tau_n$ if $WCRT_n \leq T_n$.
\begin{subequations}
  \label{eq:MILP-v1}{\small
  \begin{align}
&    \mbox{\bf maximize: } \;\;\;\;f_m \label{eq:objective-v1}\\
 &   \mbox{\bf subject to:} \nonumber\\
&\qquad\qquad\qquad    r_{i,j} \geq g_j,  \forall i=1,\ldots,n-1, \forall
     j=1,\ldots,m \label{eq:r-versus-g}\\
& \qquad\qquad\qquad   N_{i,j} \mbox{ is an integer },
    \forall i=1,\ldots,n-1, \forall j=1,\ldots,m \label{eq:N-integer-v1}\\
 & \qquad\qquad\qquad   \mbox{ and  Conditions in Eqs.~\eqref{eq:N-ij-1},~\eqref{eq:N-ij-2},~\eqref{eq:f_j},~\eqref{eq:g_j} hold}.\nonumber
  \end{align}}
\end{subequations}
\end{theorem}
\begin{proof}
  This comes from the above discussions and
  Lemma~\ref{lemma:necessary-critical-wcrt}.  The release pattern that
  has the maximum $f_m$ (provided that $g_1$ is set to $0$) by using
  FP preemptive scheduling under all the constraints due to the three
  conditions in Lemma~\ref{lemma:necessary-critical-wcrt} leads to the
  worst-case response time if $WCRT_n \leq T_n$.  
\end{proof}

However, the MILP in Eq.~\eqref{eq:MILP-v1} is not an exact response
time analysis (or schedulability test) due to the following reason:
the condition in Eq.~\eqref{eq:f_j} is only a safe upper bound on $f_j$, but does not provide the exact
$f_j$ under the release pattern. Suppose that $n$ is $2$. We have
$T_1=4$ and $C_1=2$. Consider $g_1=0$ and $r_{i,1}=0$, $C_n^1 = 2$,
and $S_n^1=8$. In this case, it implies that the suspension interval
$S_n^1$ has no impact when we analyze the worst-case finishing time of
the first computation segment.\footnote{This is also proved in
  Property \ref{property:p2}.}  It is clear that the exact
(worst-case) finishing time of $C_n^1$ is $4$ under this release
pattern. However, there is another feasible solution that satisfies
Eq.~\eqref{eq:f_j} by setting $N_{1,1}$ to $2$, $f_1$ to $6$, and
$r_{1,2}$ to $14$. Therefore, in fact, $f_1$ can have the following
cases:\footnote{For this case, it becomes infeasible when $N_{1,1}$ is
  larger than $2$.}
\begin{itemize}
\item  $f_1$ is $2$ when $N_{1,1}$ is $0$, 
\item  $f_1$ is $4$ when $N_{1,1}$ is $1$, and
\item  $f_1$ is $6$ when $N_{1,1}$ is $2$.
\end{itemize}
However, due to the objective function for \emph{maximization}, the
optimal MILP solution is to set $f_1$ to $6$ instead of $4$ under this
MILP.\footnote{This also explains why the statement in the earlier version
  of this report (\url{https://arxiv.org/abs/1605.00124v1}) was erroneous since it skipped the above discussion and directly concluded that the MILP returns the exact worst-case response time.}

\subsection{Connection to the MILP by Nelissen et al. in ECRTS 2015}

The MILP in Eq.~\eqref{eq:MILP-v1} looks different from the MILP in 
Eq.~\eqref{eq:MILP}, but they are in fact equivalent.  Suppose that
$R_j = f_j - g_j, \forall j=1,2,\ldots,m$ and $O_{i,j} = r_{i,j} -
g_j, \forall i=1,2,\ldots,n-1, \forall j=1,2,\ldots,m$. We can
rephrase the MILP in Eq.~\eqref{eq:MILP-v1} into the MILP in
Eq.~\eqref{eq:MILP} as follows:
\begin{itemize}
\item Clearly, the objective function in Eq.~\eqref{eq:objective-v1}
  is identical to that in Eq.~\eqref{eq:objective}.
\item The condition in Eq.~\eqref{eq:f_j} leads to Eq.~\eqref{eq:R_j}.
\item The condition in Eq.~\eqref{eq:r-versus-g} is identical to Eq.~\eqref{eq:O-positive}.
\item The condition in Eq.~\eqref{eq:g_j} and Eq.~\eqref{eq:f_j} can be used to rephrase Eq.~\eqref{eq:N-ij-2} into 
{\small
  \begin{equation*}
    N_{i,j} \leq \frac{r_{i,j+1} - r_{i,j}}{T_i} = \frac{g_j + R_j + S_n^j + O_{i,j+1} - (g_j + O_{i,j})}{T_i} = \frac{R_j + S_n^j + O_{i,j+1} - O_{i,j}}{T_i},
  \end{equation*}} which is identical to the condition in Eq.~\eqref{eq:O-distance}.
\item Moreover, the condition in Eq.~\eqref{eq:N-ij-1} is identical to Eq.~\eqref{eq:N-bounds}.
\end{itemize}

Therefore, we have the following corollaries.
\begin{corollary}
  The optimal solution of the MILP in Eq.~\eqref{eq:MILP} (even by excluding Eqs.~\eqref{eq:total-segment-upperbound}, \eqref{eq:total-upperbound},~or~\eqref{eq:Rj>rel-original}) is a safe
  upper bound of the 
  worst-case response time of task $\tau_n$ if $WCRT_n \leq T_n$.
\end{corollary}
\begin{corollary}
  If the optimal solution of the MILP in Eq.~\eqref{eq:MILP} (even by excluding Eqs.~\eqref{eq:total-segment-upperbound}, \eqref{eq:total-upperbound},~or~\eqref{eq:Rj>rel-original}), or
  equivalently the MILP in Eq.~\eqref{eq:MILP-v1} is no more than 
  $T_n$, then $WCRT_n \leq T_n$.
\end{corollary}

\section{Response Time Analysis: How Far is the Gap?}
\label{sec:gap}

Since the MILP approach listed in Section~\ref{sec:proof-milp} does
 not provide the exact worst-case response
time of task $\tau_n$, it is also meaningful to examine whether the
upper bound on the worst-case response time by using the MILP
approach in Section~\ref{sec:proof-milp} is always very close to (or
not too far from) the exact worst-case response time. Unfortunately,
we will demonstrate a task set, in which the derived worst-case
response time from the MILP in Eq.~\eqref{eq:MILP} is at least $\frac{4m+4}{9}$ times
the exact worst-case response time, where $m \geq 2$ is the number of
computation segments of task $\tau_n$. 

We consider the following task set ${\bf T}^{MILP}$ with $n=m+4$ tasks,
where $q$ is a positive integer, $m$ is a positive integer with $m
\geq 2$, and $0 < \epsilon < 1/q$:
\begin{itemize}
\item For task $\tau_1$, we set $C_1 = 1, S_1 = 0, D_1=T_1=2$.
\item For task $\tau_2$, we set $C_2 = q, S_2 = 0, D_2=T_2=4q$.
\item For task $\tau_3$, we set $C_3 = 2q-1+\epsilon, S_3 = 0, D_3=T_3=8q$.
\item For task $\tau_i$ with $i=4,5,\ldots, m+3$, we set $C_i = 1-\epsilon,
  S_i = 0, D_i=8qm, T_i= 16qm^2+(m-1)(2q-1)$.
\item For task $\tau_{m+4}$, we create a segmented self-suspending
  task with $m$ computation segments separated by $m-1$
  self-suspension intervals, i.e., $m_{n}=m$, in which $C_{m+4}^j=
  1-\epsilon$ for $j=1,2,\ldots,m$, $S_{m+4}^j = 2q-1$ for
  $j=1,2,\ldots,m-1$. The values of $D_{m+4}$ and $T_{m+4}$ are left
  open, and our goal here is to find the minimum feasible $D_{m+4}$
  that can be set when $T_{m+4}$ is large enough.
\end{itemize}

The following property is very useful when we need to calculate the
worst-case response time:

\begin{property}
\label{property:upper-bound-exact-v2}
  For a given positive integer $x$, the minimum $t | t > 0$ such that
  $x(1-\epsilon)+\ceiling{\frac{t}{2}} + \ceiling{\frac{t}{4q}}q
  +\ceiling{\frac{t}{8q}}(2q-1+\epsilon) =t$ happens when $t$ is
  $x\cdot 8q$.
\end{property}
\begin{proof}
  This can be proved by simple arithmetics.  
\end{proof}
By using Property~\ref{property:upper-bound-exact-v2}, it is not
difficult to obtain the exact worst-case response time by using
Lemma~\ref{lemma:necessary-critical-wcrt}.
\begin{lemma}
  \label{lemma:task-MILP}
  Tasks $\tau_1, \tau_2, \ldots, \tau_{m+3}$ in ${\bf T}^{MILP}$ can
  meet their deadlines. The worst-case response time of task $\tau_n$
  in ${\bf T}^{MILP}$ is $16qm+(m-1)(2q-1)$.
\end{lemma}
\begin{proof}
  The schedulability of tasks $\tau_1, \tau_2, \tau_3$ comes by using
  the standard time demand analysis, and the schedulability of tasks
  $\tau_4, \tau_5, \ldots, \tau_{m+3}$ follows from
  Property~\ref{property:upper-bound-exact-v2}.  The constructed task
  set ${\bf T}^{MILP}$ has the following properties based on
  Lemma~\ref{lemma:necessary-critical-wcrt}: a) We should always
  release the three highest priority tasks together with a computation
  segment of task $\tau_n$ and release their subsequent jobs
  periodically and as early as possible by respecting their minimum
  inter-arrival times until this computation segment of task $\tau_n$
  finishes. b) If the response time of task $\tau_n$ is no more than
  $16qm^2+(m-1)(2q-1)$, then each task $\tau_i$ for $i=4,5,\ldots,m+3$
  only releases one job to interfere in a computation segment of task
  $\tau_n$.

  Suppose that there are $\ell_j$ tasks among $\tau_4, \tau_5, \ldots,
  \tau_{m+3}$ which interfere in the $j$-th computation segment of task
  $\tau_n$. We know that $\ell_j$ is an integer and $\ell_j \geq 0$
  for $j=1,2,\ldots,m$ and $\sum_{j=1}^{m} \ell_j = m$. Moreover, the
  response time of $j$-th computation segment of task $\tau_n$ is
  $(\ell_j+1)\cdot 8q$ by
  Property~\ref{property:upper-bound-exact-v2}.  Therefore, we know
  that the worst-case response time of task $\tau_n$ is
 \[
   \left(\sum_{j=1}^{m} (\ell_j+1)\cdot 8q\right) + (m-1)(2q-1) = 16qm + (m-1)(2q-1).
 \]
 Since $16qm+(m-1)(2q-1) < T_i$ for $i=4,5,\ldots,n-1$, we know that
 the above value is an upper bound by all the possible release
 patterns that satisfy Lemma~\ref{lemma:necessary-critical-wcrt}. And,
 there is a concrete release/execution pattern which leads the
 response time of task $\tau_n$ exactly to this upper
 bound. Therefore, this is the exact worst-case response time of task
 $\tau_n$ in ${\bf T}^{MILP}$.
% First, we notice that $T_1-C_1 = S_n^j$ for any
%   $j=1,2,\ldots,m-1$, where $n$ is $m+2$. Therefore,
%   Property~\ref{property:p4} can be used for task $\tau_1$ in ${\bf
%     T}^{MILP}$. Moreover, since $T_i$ is equal to $T_n$ for
%   $i=2,3,\ldots,m+1$, we can use Property \ref{property:p3} for such
%   medium priority tasks. That is, each task $\tau_i$ for
%   $i=2,3,\ldots,m+1$ only releases one job together with a computation
%   segment of task $\tau_n$. Suppose that there are $\ell_j$ tasks
%   among $\tau_2, \tau_3, \ldots, \tau_{m+1}$ interfere in the $j$-th
%   computation segment of task $\tau_n$. We know that $\ell_j$ is an
%   integer and $\ell_j \geq 0$ for $j=1,2,\ldots,m$ and $\sum_{j=1}^{m}
%   \ell_j = m$. Moreover, the response time of $j$-th computation
%   segment of task $\tau_n$ is $(\ell_j+1)(q+1)$ by
%   Property~\ref{property:upper-bound-exact}, since $\ell_j$ is a
%   non-negative integer. Therefore, we know that the worst-case
%   response time of task $\tau_n$ is
% \[
%   \left(\sum_{j=1}^{m} (\ell_j+1)(q+1)\right) + m(m-1) = 2m(q+1) + m(m-1).
% \]
% That is, for task set ${\bf T}^{MILP}$, the release patterns of the
% $m$ medium-priority tasks have no influence on the worst-case response
% time of task $\tau_n$ (as long as task $\tau_i$ for $i=2,3,\ldots,
% m+1$ in fact interferes in exactly one computation segment of task
% $\tau_n$).
\end{proof}

\subsection{Excluding the Boundary Constraints by Eq.~\eqref{eq:total-segment-upperbound} and Eq. \eqref{eq:total-upperbound}}

We first investigate whether the MILP without the boundary constraints presented by Eq.~\eqref{eq:total-segment-upperbound} and Eq. \eqref{eq:total-upperbound}.  We explore this specific condition under a special case, by further ignoring
 the interference of the tasks $\tau_4, \tau_5, \ldots,
 \tau_{m+3}$. Then, the worst-case response time of the $j$-th
 computation segment of task $\tau_n$ (after the segment is released)
 can be obtained by the following MILP.
\begin{subequations}
  \label{eq:MILP-v3}{\small
  \begin{align}
   & \mbox{\bf maximize: } R_j \label{eq:objective-v3}\\
   & \mbox{\bf subject to:} \nonumber\\
    & R_j = 1- \epsilon + N_{1,j} \cdot 1 + N_{2,j} \cdot q + N_{3,j}
     \cdot (2q-1+\epsilon), \qquad\label{eq:R-v3}\\
    & O_{1,j} \geq 0, O_{2,j} \geq 0, O_{3,j} \geq 0 \qquad \label{eq:O-v3}\\
     &0 \leq N_{i,j}  \leq \ceiling{\frac{R_j - O_{i,j}}{T_i}},\qquad
    \forall i=1,\ldots,3 \label{eq:N-v3}\\
    & R_j > rel_{i,j} + \sum_{\ell=1}^{3}\max\left\{0,
     \floor{\frac{O_{\ell,j}+N_{\ell,j}T_\ell - rel_{i,j}}{T_\ell}}C_\ell\right\}, \qquad
    \forall i=1,\ldots,3\label{eq:R>rel-v3}\\
    & rel_{i,j} = O_{i,j} + (N_{i,j}-1) T_i,  \qquad
    \forall i=1,\ldots,3\label{eq:rel-v3}\\
    & N_{i,j} \mbox{ is an integer },\qquad
    \forall i=1,\ldots,3\label{eq:N-integer-v3}
  \end{align}}
\end{subequations}

\begin{lemma}
\label{lemma:upper-bound-pessmistic-v3}
By the assumption that $q$ is a positive integer $q \geq 1$ and $0 < \epsilon < 1/q$, the setting of
$R_j=8q^2+6q+1+q\epsilon$, $O_{1,j} = 0$, $O_{2,j} = \epsilon/4$, $O_{3,j} =
\epsilon/2$, $N_{1,j} = 4q^2+3q+1$, $N_{2,j} = 2q+2$, and $N_{3,j}
= q+1$ is a feasible solution of the MILP in Eq.~\eqref{eq:MILP-v3}.
\end{lemma}
\begin{proof}
  The first condition in Eq.~\eqref{eq:R-v3} holds since
{\small\begin{align*}
& 1-\epsilon+4q^2+3q+1+(2q+2)\cdot q + (q+1)\cdot(2q-1+\epsilon)\\ 
= &1-\xcancel{\epsilon}+4q^2+3q+\xcancel{1}+2q^2+2 q + 2q^2 -q + q\epsilon +2q-\xcancel{1} + \xcancel{\epsilon}\\
= & 8q^2+6q+1+q\epsilon.  
\end{align*}}
The conditions in Eqs.~\eqref{eq:O-v3},~\eqref{eq:N-v3},
and~\eqref{eq:N-integer-v3} clearly hold.  In this case, the condition
in Eq.~\eqref{eq:rel-v3} sets $rel_{1,j} = O_{1,j} + (N_{1,j}-1)\times
2 = 8q^2+6q$, $rel_{2,j} = O_{2,j} + (N_{2,j}-1)\times 4q =
8q^2+4q+\epsilon/4$, and $ rel_{3,j} = O_{3,j} + (N_{3,j}-1)\times 8q
= 8q^2 + \epsilon/2$. Now, we verify whether the condition in
Eq.~\eqref{eq:R>rel-v3} holds:
  \begin{itemize}
  \item When $i=1$, we have 
\[
8q^2+6q + 1 < R_j.
\]
  \item When $i=2$, we have 
    \begin{align*}
   &   8q^2+4q + \epsilon/4 + \max\left\{0, \floor{\frac{8q^2+6q+2 -
          (8q^2+4q + \epsilon/4)}{2}}\right\}\\
   &+  q + \max\left\{0, \floor{\frac{8q^2+8q+\epsilon/2 -
          (8q^2+4q + \epsilon/4)}{8q}}(2q-1+\epsilon)\right\}\\
=\;\;\; & 8q^2+4q + \epsilon/4  + q+q+0  =8q^2+6q+\epsilon/4< R_j 
    \end{align*}
  \item When $i=3$, we have 
    \begin{align*}
   &   8q^2+ \epsilon/2 + \max\left\{0, \floor{\frac{8q^2+6q+2 -
          (8q^2 + \epsilon/2)}{2}}\right\}\\
   &+  \max\left\{0, \floor{\frac{8q^2+8q+\epsilon/4 -
          (8q^2 + \epsilon/2)}{4q}}q\right\} + 2q+1-\epsilon\\
= & 8q^2+ \epsilon/2  + 3q+q+2q+1-\epsilon  = 8q^2+6q+1-\epsilon/2< R_j 
    \end{align*}
  \end{itemize}
  Therefore, we reach the conclusion.
\end{proof}

Now, we can examine the MILP in Eq.~\eqref{eq:MILP}, when
excluding 
Eq.~\eqref{eq:total-segment-upperbound} and
Eq. \eqref{eq:total-upperbound}:
\begin{subequations}
  \label{eq:MILP-v4}{\small
  \begin{align}
  &  \mbox{\bf maximize: } \;\;\;\;S_n + \sum_{j=1}^{m} R_j \label{eq:objective-v4}\\
   & \mbox{\bf subject to:} \qquad\qquad\qquad&\nonumber\\
  &   R_j = C_n^j + \sum_{i=1}^{n-1} N_{i,j} \times C_i, \qquad
     \forall j=1,\ldots,m \label{eq:R_j-v4}\\
    & O_{i,j} \geq 0, \qquad \forall i=1,\ldots,n-1, \forall
     j=1,\ldots,m \label{eq:O-positive-v4}\\
    & O_{i,j+1} \geq O_{i,j} + N_{i,j}\times T_i - (R_j + S_n^j), \qquad \label{eq:O-distance-v4}
    \forall i=1,\ldots,n-1, \forall j=1,\ldots,m-1\\
    & 0 \leq N_{i,j}  \leq \ceiling{\frac{R_j - O_{i,j}}{T_i}}, \qquad
    \forall i=1,\ldots,n-1, \forall j=1,\ldots,m \label{eq:N-bounds-v4}\\
   & R_j > rel_{i,j} + \sum_{\ell=1}^{n-1}\max\left\{0,
     \floor{\frac{O_{\ell,j}+N_{\ell,j}T_\ell - rel_{i,j}}{T_\ell}}C_\ell\right\} \;\;
    \forall i=1,\ldots,n-1, \forall j=1,\ldots,m\label{eq:R>rel-v4}\\
    & rel_{i,j} = O_{i,j} + (N_{i,j}-1) T_i  \qquad
    \forall i=1,\ldots,n-1, \forall j=1,\ldots,m\label{eq:rel-v4}\\
    & N_{i,j} \mbox{ is an integer }, \qquad
    \forall i=1,\ldots,n-1, \forall j=1,\ldots,m \label{eq:N-integer-v4}
  \end{align}}
\end{subequations}

\begin{lemma}
\label{property:upper-bound-pessmistic-v4}
Suppose that $q$ is a positive integer $q \geq 1$ and $0 < \epsilon <
1/q$.  For any $j=1,2,\ldots,m$, the setting of
$R_j=8q^2+6q+1+q\epsilon$, $O_{1,j} = 0$, $O_{2,j} = \epsilon/4$,
$O_{3,j} = \epsilon/2$, $N_{1,j} = 4q^2+3q+1$, $N_{2,j} = 2q+2$, and
$N_{3,j} = q+1$ is a feasible solution of the MILP in
Eq.~\eqref{eq:MILP-v4} for ${\bf T}^{MILP}$ when $N_{i,j}=0, O_{i,j} = 0$ for all
$i=4,5,\ldots,n-1$. Therefore, the optimal solution of
Eq.~\eqref{eq:MILP-v4} is at least $S_n + m( 8q^2+6q+1+q\epsilon) =
(m-1)(2q-1) + m( 8q^2+6q+1+q\epsilon)$.
\end{lemma} 
\begin{proof}
  By Lemma~\ref{lemma:upper-bound-pessmistic-v3}, we only need
  to further verify whether the condition in
  Eq.~\eqref{eq:O-distance-v4} holds when $i=1,2,3$. By the definition
  $S_n^j = (2q-1)$, we know that $O_{i,j} + N_{i,j} \times T_i -
  (R_j+S_n^j) = O_{i,j} + N_{i,j} \times T_i - (8q^2+8q+q\epsilon) <
  0$ for $i=1,2,3$. Therefore, the condition in
  Eq.~\eqref{eq:O-distance-v4} is by definition satisfied.
\end{proof}

Now, we can reach the conclusion that MILP in Eq.~\eqref{eq:MILP-v4}
can be very far from the exact worst-case response time by the
following theorem.
\begin{theorem}
\label{theorem-unbounded-q}
  The result of the MILP in Eq.~\eqref{eq:MILP-v4} for task $\tau_n$
  in task set ${\bf T}^{MILP}$ divided by the exact worst-case
  response time of task $\tau_n$ is at least $\frac{
    m(8q^2+6q+1+q\epsilon) + (m-1)(2q-1)}{16qm+(m-1)(2q-1)}$.  The
  ratio can become unbounded by the number of computation segments or
  the number of tasks when $q$ is sufficiently large.
\end{theorem}
\begin{proof}
  This follows directly from Lemmas~\ref{lemma:task-MILP} and
  \ref{property:upper-bound-pessmistic-v4}.
\end{proof}

\begin{corollary}
\label{corollary-unbounded-q}
The result of the MILP in Eq.~\eqref{eq:MILP} by excluding the
boundary constraints presented by Eq.~\eqref{eq:total-segment-upperbound} and
Eq. \eqref{eq:total-upperbound} for task $\tau_n$ in task set ${\bf
  T}^{MILP}$ divided by the exact worst-case response time of task
$\tau_n$ is at least $\frac{ m(8q^2+6q+1+q\epsilon) +
  (m-1)(2q-1)}{16qm+(m-1)(2q-1)}$.  
\end{corollary}
\begin{proof}
  This follows directly from Theorem~\ref{theorem-unbounded-q}.
\end{proof}

\subsection{Improvements by the Boundary Conditions
  Eq.~\eqref{eq:total-segment-upperbound} and Eq.~\eqref{eq:total-upperbound}}

% From the above discussions, including only Eq. \eqref{eq:Rj>rel-original} (i.e., Eq. (16) in
% \cite{ecrts15nelissen}) may not be sufficient to reduce the pessimism
% due to Theorem~\ref{theorem-unbounded-q}. 

% Let us further examine the
% usefulness of Eq. (9) and Eq. (11) in the MILP in
% \cite{ecrts15nelissen}, both are based on the boundary conditions from
% existing analytical methods. Eq. (9) in the MILP in
% \cite{ecrts15nelissen} in the notation system in this report is
% identical as follows,
% \begin{equation}
%   \label{eq:total-upperbound}
%   S_n + \sum_{j=1}^{m} R_j \leq UB_{ss},
% \end{equation}
% and Eq. (11) in the MILP in \cite{ecrts15nelissen} is identical as follows,
% \begin{equation}
%   \label{eq:total-segment-upperbound}  
%   R_j \leq UB_{ss,j} \qquad \forall j=1,2,\ldots,m,
% \end{equation}
% where $UB_{ss}$ is defined
% as the upper bound on the worst-case response time of task $\tau_n$,
% and $UB_{ss,j}$ is defined as the upper bound on the worst-case
% response time of the $j$-th computation segment of task
% $\tau_n$. 

We now discuss the complete MILP in Eq.~\eqref{eq:MILP}.  Calculating
$UB_{ss,j}$ for task $\tau_n$ in ${\bf T}^{MILP}$ is rather
straightforward. This can be done by releasing all the jobs together
with $C_n^j$. That is, $UB_{ss,j}$ is the minimum $t | t > 0$ such
that $(m+1)(1-\epsilon)+\ceiling{\frac{t}{2}} +
\ceiling{\frac{t}{4q}}q +\ceiling{\frac{t}{8q}}(2q-1+\epsilon)=t$. By
Property~\ref{property:upper-bound-exact-v2}, we know that $UB_{ss,j}$
is $(m+1)\cdot 8q$.

Calculating $UB_{ss}$ is tricky. However, for task $\tau_n$ in ${\bf
  T}^{MILP}$ we can easily conclude that $UB_{ss} \leq m(m+1)\cdot 8q
+ S_n = 8qm(m+1) + (m-1)(2q-1)$. If we can get a very tight upper
bound of $UB_{ss}$, then, there is no need of the MILP. Here is how
$UB_{ss}$ was proposed to be calculated by Nelissen et
al. \cite{ecrts15nelissen}:
\begin{quote}
  Nelissen et al. \cite{ecrts15nelissen}:\footnote{The text is
    reorganized to use the proper references and notation in this
    paper.} \emph{Constraints \eqref{eq:total-segment-upperbound} and
  \eqref{eq:total-upperbound} reduce the research space of the problem
  by stating that the overall response time of $\tau_n$ and the
  response time of each of its execution regions, respectively, cannot
  be larger than known upper-bounds computed with simple methods such
  as the joint and split methods presented in \cite{bletsas:thesis}.}
\end{quote}

\begin{lemma}
  \label{lemma:upper-bound-milp-full}
  When $q$ is set to $m$ and $m \geq 2$, $UB_{ss}$ derived from the
  joint and split methods presented in \cite{bletsas:thesis} is at
  least $8m^2(m+1) + (m-1)(2m-1)$ for ${\bf T}^{MILP}$.
\end{lemma}
\begin{proof}
  The joint and split methods presented in \cite[Pages
  131-141]{bletsas:thesis} are based on the following concept:
  \begin{itemize}
  \item A self-suspension interval of task $\tau_n$ can be converted
    to computation demand. (joint)
  \item A self-suspension interval of task $\tau_n$ can be treated as
    self-suspension, by considering their suffered worst-case
    interference independently. (split)
  \end{itemize}
  
  The following proof is only sketched since this can be easily proved
  by a simple observation. If we consider a self-suspension interval
  $S_n^j$ as computation (i.e., the joint approach), then, the
  additional workload $(2q-1)$ (due to suspension as computation)
  increases the worst-case response time by $(2q-1)8q = (2m-1)8m =
  16m^2-8m$.  If we simply treat these two consecutive computation
  segments $C_n^{j-1}$ and $C_n^{j}$ by considering their suffered
  worst-case interference independently (i.e., the split approach),
  this treatment only increases the worst-case response time by at
  most $8m^2+2m-1$. Therefore, the \emph{joint} approach is always
  worse than the \emph{split} approach, when $m \geq 2$.  This can be
  formally proved by starting from $j=1$ to convert any joint
  treatment to a split treatment in a stepwise manner.

  Hence, $UB_{ss} = m(m+1)\cdot 8q + S_n = 8m^2(m+1) + (m-1)(2m-1)$ by
  splitting all the computation segments.
\end{proof}

With the above
discussions, we can reach the following lemma:

\begin{lemma}
  \label{lemma:lower-bound-milp-full}
  When $q$ is set to $m$ and $m \geq 2$, the objective function of the
  MILP in Eq.~\eqref{eq:MILP} for task $\tau_n$ in ${\bf T}^{MILP}$ is
  at least $(m-1)(2m-1) + m( 8m^2+6m+1+m\epsilon)$.
\end{lemma}
\begin{proof}
  Since $(m+1)\cdot 8q > 8m^2+6m+1+m\epsilon$ when $q$ is set to $m$,
  we know that Eq.~\eqref{eq:total-segment-upperbound} is satisfied by
  adopting the solution in
  Lemma~\ref{property:upper-bound-pessmistic-v4}. Similarly, since
  $8qm(m+1) + (m-1)(2q-1) > (m-1)(2m-1) + m( 8m^2+6m+1+m\epsilon)$
  when $q$ is set to $m$, we also know that
  Eq.~\eqref{eq:total-upperbound} is satisfied by adopting the
  solution in
  Lemma~\ref{property:upper-bound-pessmistic-v4}. Therefore, the
  solution in Lemma~\ref{property:upper-bound-pessmistic-v4} is a
  feasible solution of the MILP in Eq.~\eqref{eq:MILP}, in which we reach the conclusion of
  this lemma.
\end{proof}

\begin{theorem}
\label{theorem-bounded-M}
The result of the MILP in Eq.~\eqref{eq:MILP} (i.e., the MILP in
\cite{ecrts15nelissen}) for task $\tau_n$ in ${\bf T}^{MILP}$ divided
by the exact worst-case response time of task $\tau_n$ is at least\\
$\frac{ m(8m^2+6m+1+m\epsilon) + (m-1)(2m-1)}{16m^2+(m-1)(2m-1)} \geq
\frac{4m+4}{9}$, when $m \geq 2$.
\end{theorem}
\begin{proof}
  This follows directly from Lemmas~\ref{lemma:task-MILP} and
  \ref{lemma:lower-bound-milp-full}, and
  \begin{align*}
& \frac{ m(8m^2+6m+1+m\epsilon) + (m-1)(2m-1)}{16m^2+(m-1)(2m-1)}
= \frac{8m^3+8m^2-2m+1+ m^2\epsilon}{18m^2-3m+1} \\
= \;\;&\frac{4m+4}{9} + \frac{\frac{4m^2}{3} -\frac{10m}{9} + \frac{5}{9} +
  m^2\epsilon}{18m^2-3m+1} > \frac{4m+4}{9}.
  \end{align*}
\end{proof}

\section{Conclusions and Discussions}

This report shows that the schedulability analysis for fixed-priority
preemptive scheduling even with only one segmented self-suspending
task as the lowest-priority task is $co{\cal NP}$-hard in the strong
sense. Moreover, we also show that the upper bound on the worst-case
response time by using a mixed-integer linear programming (MILP)
formulation by Nelissen et al. \cite{ecrts15nelissen} can be at least
$\Omega(m)$ times the exact worst-case response time, where $m$ is the
number of computation segments of task $\tau_n$.

Therefore, how to analyze the worst-case response time tightly remains
as an open problem for self-suspending sporadic task systems even with
one self-suspending sporadic task as the lowest-priority task under
fixed-priority preemptive scheduling.

\vspace{0.2in}
{\bf Acknowledgements.}  The author would like to thank
Dr. Geoffrey Nelissen from CISTER, ISEP, Polytechnic Institute of
Porto and Prof. Dr. Cong Liu from UT Dallas for their feedbacks on an
earlier version of this report, which help the author improve the presentation
flow and the clarity.  This report is supported by DFG, as part of the
Collaborative Research Center SFB876 (http://sfb876.tu-dortmund.de/).

{\small
  \bibliography{real-time,biblio,biblio-summary}{} }
\end{document}